\documentclass{stacs_proc}
\usepackage{epsfig}

\newcommand{\sort}{\mathrm{sort}}
\newcommand{\scan}{\mathrm{scan}}

\begin{document}

\title{On Dynamic Breadth-First Search in External-Memory}
\author[lab1]{U. Meyer}{Ulrich Meyer}

\thanks{Partially supported by the DFG grant ME 3250/1-1, and by the center of massive data algorithmics (MADALGO) funded by the Danish National Research Foundation.}

\address[lab1]{Institute for Computer Science\\
J.~W.~Goethe University\\
60325 Frankfurt/Main, Germany}

\email{umeyer@ae.cs.uni-frankfurt.de}


\begin{abstract}
We provide the first non-trivial result on dynamic breadth-first search (BFS) in external-memory:
For general sparse undirected graphs of initially $n$ nodes and $O(n)$ edges
and monotone update sequences of either 
$\Theta(n)$ edge insertions or $\Theta(n)$ edge deletions, we prove an amortized 
high-probability bound of $O(n/B^{2/3}+\sort(n)\cdot \log B)$ I/Os per update. In contrast, the currently
best approach for static BFS on sparse undirected graphs requires 
$\Omega(n/B^{1/2}+\sort(n))$ I/Os.

\end{abstract}

\maketitle

\stacsheading{2008}{551-560}{Bordeaux}
\firstpageno{551}

\section{Introduction}
Breadth first search (BFS) is a fundamental graph traversal strategy. It can also be viewed as computing single source shortest paths on unweighted graphs. It decomposes the input graph $G=(V,E)$ of $n$ nodes and $m$ edges into at most $n$ levels where level~$i$ comprises all nodes that can be reached from a designated source~$s$ via a path of $i$~edges, but cannot be reached using less than $i$ edges. 

The objective of a dynamic graph algorithm is to efficiently process an online sequence of
update and query operations; see \cite{EppGalIta-ATCH-99,Roditty:2006} 
for overviews of classic and recent results.
In our case we consider BFS under a sequence of either $\Theta(n)$ edge insertions, 
but not deletions  ({\em incremental} version) or  $\Theta(n)$ edge deletions, 
but not insertions ({\em decremental} version). After each edge insertion/deletion
the updated BFS level decomposition has to be output.

\subsection{Computation models.}
We consider the commonly accepted external-memory (EM) model of Aggarwal and Vitter \cite{AGG_VIT}. It assumes a two level memory hierarchy with faster internal memory having a capacity to store $M$ vertices/edges. In an I/O operation, one block of data, which can store $B$ vertices/edges, is transferred between disk and internal memory.  The measure of performance of an algorithm is the number of I/Os it performs.  The number of I/Os needed to read $N$ contiguous items from disk is $\scan(N) = \Theta(N/B)$.  The number of I/Os required to sort $N$ items is $\sort(N)=\Theta((N/B)\log_{M/B}(N/B))$. For all realistic values of $N$, $B$, and $M$, $\scan(N)<\sort(N)\ll N$. 

There has been a significant number of publications on external-memory graph
algorithms; see~\cite{Alg_Mem,VIT1} for recent overviews. However, we are not aware of
any dynamic graph algorithm in the fully external-memory case (where
$|V| > M$).

\subsection{Results.}
We provide the first non-trivial result on dynamic BFS in external-memory.
For general sparse undirected graphs of initially $n$ nodes and $O(n)$ edges
and either $\Theta(n)$ edge insertions or $\Theta(n)$ edge deletions, we prove an amortized 
high-probability bound of $O(n/B^{2/3}+\sort(n)\cdot \log B)$ I/Os per update. 
In contrast, the currently
best bound for static BFS on sparse undirected graphs is 
$O(n/B^{1/2}+\sort(n))$ I/Os~\cite{MM_BFS}. 

Also note that for general sparse graphs and  worst-case monotone sequences of $\Theta(n)$ updates 
in {\em internal-memory} there is asymptotically no better solution 
than performing $\Theta(n)$ runs of the linear-time static BFS algorithm,
even if after each update we are just required to report the changes in the BFS tree
(see Fig.~\ref{bezug} for an example).
In case $\Omega(n/B^{1/2}+\sort(n))$ I/Os should prove to be a lower bound for static
BFS in external-memory, then our result yields an interesting differentiator
between static vs. dynamic BFS in internal and external memory.
\begin{figure}[htpb]
\begin{center}
\epsfig{figure=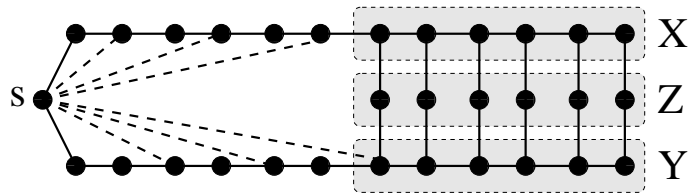,height=2.5cm,angle=0}
\caption{\label{bezug}Example for a graph class where each update requires $\Omega(n)$ changes
in the BFS tree: inserting new (dashed) edges alternatingly shortcut the distances 
from $s$ to {\bf X} and $s$ to {\bf Y}. As a result, in the updated BFS tree
the parents of all vertices in {\bf Z} keep on changing between {\bf X} and~{\bf Y}.}
\end{center}
\end{figure} 
\subsection{Organization of the paper.}
In Section~\ref{static} we will review known BFS algorithms for static undirected graphs.
Then we consider traditional and new external-memory methods for graph clustering (Section~\ref{clusters}). Subsequently, in Section~\ref{algorithm} we provide the
new algorithm and analyze it in Section~\ref{analysis}. Final remarks concerning extensions
and open problems are given in Sections~\ref{extensions} and~\ref{conclusions}, respectively.

\section{Review of Static BFS Algorithms}
\label{static}
\noindent{\bf Internal-Memory.} BFS is well-understood in the RAM model. There exists a simple linear time algorithm \cite{Cormen} (hereafter referred as IM\_BFS) for the BFS traversal in a graph. IM\_BFS keeps a set of appropriate candidate nodes for the next vertex to be visited in a FIFO queue $Q$. Furthermore, in order to find out the unvisited neighbors of a node from its adjacency list, it marks the nodes as either visited or unvisited. 
\par Unfortunately, as the storage requirements of the graph starts approaching the size of the internal memory, the running time of this algorithm deviates significantly from the predicted $O(n+m)$ asymptotic performance of the RAM model: checking whether edges lead to already visited nodes altogether needs $\Theta(m)$ I/Os in the worst case; unstructured indexed access to adjacency lists may add another $\Theta(n+m/B)$ I/Os.\\

\noindent {\bf EM-BFS for dense undirected graphs.} The algorithm by Munagala and Ranade \cite{MR_BFS} (referred as MR\_BFS) ignores the second problem but addresses the first by exploiting the fact that the neighbors of a node in BFS level $t-1$ are all in BFS levels $t-2$, $t-1$ or $t$. Let $L(t)$ denote the set of nodes in BFS level $t$, and let $A(t)$ be the multi-set of neighbors of nodes in $L(t-1)$. Given $L(t-1)$ and $L(t-2)$, MR\_BFS builds $L(t)$ as follows: Firstly, $A(t)$ is created by $|L(t-1)|$ random accesses to get hold of the adjacency lists of all nodes in $L(t-1)$. Thereafter, duplicates are removed from $A(t)$ to get a sorted set $A'(t)$. This is done by sorting $A(t)$ according to node indices, followed by a scan and compaction phase. The set $L(t):=A'(t)\ \setminus \{L(t-1) \cup L(t-2)\}$ is computed by scanning ``in parallel'' the sorted sets of $A'(t), L(t-1)$, and $L(t-2)$ to filter out the nodes already present in $L(t-1)$ or $L(t-2)$. The resulting worst-case I/O-bound is $O\left(\sum_t L(t) + \sum_t \sort(A(t))\right) = O \left(n+\sort(n+m)\right)$.
The algorithm outputs a BFS-level decomposition of the vertices, which can be easily
transformed into a BFS tree using $O(\sort(n+m))$ I/Os \cite{BGVW00}.
\\

\noindent {\bf EM-BFS for sparse undirected graphs.}
Mehlhorn and Meyer suggested another approach \cite{MM_BFS} (MM\_BFS) which involves a preprocessing phase to restructure the adjacency lists of the graph representation. It groups the vertices of the input graph into disjoint clusters of small diameter in $G$ and stores the adjacency lists of the nodes in a cluster contiguously on the disk. Thereafter, an appropriately modified version of MR\_BFS is run. MM\_BFS exploits the fact that whenever the first node of a cluster is visited then the remaining nodes of this cluster will be reached soon after. By spending only one random access (and possibly, some sequential accesses depending on cluster size) to load the whole cluster and then keeping the cluster data in some efficiently accessible data structure (pool) until it is all processed, on sparse graphs the total amount of I/Os can be reduced by a factor of up to $\sqrt{B}$: the neighboring nodes of a BFS level can be computed simply by scanning the pool and not the whole graph. Though some edges may be scanned more often in the pool, unstructured I/Os to fetch adjacency lists is considerably reduced, thereby reducing the total number of I/Os. \\

\section{Preprocessing}
\label{clusters}

\subsection{Traditional preprocessing within MM\_BFS.} 
Mehlhorn and Meyer \cite{MM_BFS} 
proposed the algorithms MM\_BFS\_R and MM\_BFS\_D, out of which the first is randomized and
the second is deterministic. 
In MM\_BFS\_R, the partitioning is generated ``in parallel rounds'':
after choosing master nodes independently and uniformly at random,
in each round, each master node tries to capture all unvisited neighbors of its current sub-graph into its partition, with ties being resolved arbitrarily. 

A similar kind of randomized preprocessing is also applied in parallel~\cite{Ullman91} and
streaming~\cite{Tradeoff_Streaming} settings. There, however, a dense compressed
graph among the master nodes is produced, causing rather high parallel work or
large total streaming volume, respectively.

The MM\_BFS\_D variant first builds a spanning tree $T_s$ for the connected component of $G$ that contains the source node. Arge et al.~\cite{Ext_MST} show an upper bound of 
$O((1+\log{\log{(B \cdot n/m)}})\cdot \sort(n+m))$ I/Os for computing such a spanning tree.
Each undirected edge of $T_s$ is then replaced by two oppositely directed edges. Note that a bi-directed tree always has at least one Euler tour. In order to construct the Euler tour around this bi-directed tree, each node chooses a cyclic order \cite{Euler} of its neighbors. The successor of an incoming edge is defined to be the outgoing edge to the next node in the cyclic order. The tour is then broken at the source node and the elements of the resulting list are then stored in consecutive order using an external memory list-ranking algorithm; 
Chiang et al.~\cite{Chi_Ext_gr_alg} showed how to do this in sorting complexity. Thereafter, we chop the Euler tour into {\em chunks} of $\max\{1,\sqrt{\frac{n \cdot B}{n+m}}\}$ nodes and remove duplicates such that each node only remains in the first 
chunk it originally occurs; again this requires a couple of sorting steps.  
The adjacency lists are then re-ordered based on the position of their corresponding nodes in the chopped duplicate-free Euler tour: all adjacency lists for nodes in the same chunks form
a cluster and the distance in $G$ between any two vertices whose adjacency-lists belong
to the same cluster is bounded by $\max\{1,\sqrt{\frac{n \cdot B}{n+m}}\}$.

\subsection{Modified preprocessing for dynamic BFS.}
\label{ssec:modified:preprocessing}
The preprocessing methods for the static BFS in \cite{MM_BFS} may produce very unbalanced
clusters: for example, with MM\_BFS\_D using chunk size $1 < \mu < O(\sqrt{B})$ there may be 
$\Omega(n/\mu)$ clusters being in charge of only $O(1)$ adjacency-lists each.
For the dynamic version, however, we would like  to argue that each random access to a 
cluster not visited so far provides us with $\Omega(\mu)$ new adjacency-lists.
Unfortunately, finding such a clustering I/O-efficiently seems to be quite hard.
Therefore, we shall already be satisfied with an Euler tour based randomized construction 
ensuring that the {\em expected} number of adjacency-lists kept in all but 
one\footnote{The last chunk of the Euler tour only visits $((2\cdot n' -1) \bmod \mu)+1$ vertices
where $n'$ denotes the number of vertices in the connected component of the starting node~$s$.} 
clusters is $\Omega(\mu)$.

The preprocessing from MM\_BFS\_D is modified as follows: each vertex~$v$
in the spanning tree $T_s$ is assigned an independent binary random number $r(v)$ with 
$\mathbf{P}[r(v)=0]= \mathbf{P}[r(v)=1]=1/2$. When removing duplicates from the Euler
tour, instead of storing $v$'s adjacency-list in the cluster related to the chunk with
the {\em first} occurrence of a vertex~$v$, now we only stick to its
first occurrence iff $r(v)=0$ and otherwise ($r(v)=1$) store $v$'s adjacency-list
in the cluster that corresponds to the {\em last} chunk of the Euler tour $v$ appears in.
For leaf nodes $v$, there is only one occurrence on the tour, hence the value of $r(v)$
is irrelevant. Obviously, each adjacency-lists is stored only once. Furthermore, 
the modified procedure maintains all good properties of the standard preprocessing
within MM\_BFS\_D like guaranteed bounded distances of $O(\mu)$ in $G$ between the vertices 
belonging to the same cluster and $O(n/\mu)$ clusters overall.   

\begin{lemma}
For chunk size $\mu>1$ and each but the last chunk, the expected number of adjacency-lists 
kept is at least $\mu/8$.
\end{lemma}
\begin{proof} Let $R=(v_1, \ldots, v_{\mu})$ be the sequence of vertices visited by an
arbitrary chunk ${\mathcal R}$ of the Euler tour ${\mathcal T}$, excluding the last chunk. 
Let $a$ be the number of entries in $R$ that represent first or last visits of 
inner-tree vertices from the spanning tree $T_s$ on ${\mathcal T}$. 
These $a$ entries account for an expected
number of $a/2$ adjacency-lists actually stored and kept in~${\mathcal R}$. 
Note that if for some vertex $v \in {\mathcal T}$ 
both its first and last visit happen within ${\mathcal R}$, then $v$'s adjacency-list 
is kept with probability one.  Similarly, if there are any visits of leaf nodes 
from $T_s$ within ${\mathcal R}$, then their adjacency-lists are kept for sure; 
let $b$ denote the number of these leaf node entries in $R$. 
What remains are $\mu-a-b$ {\em intermediate} (neither first nor last) visits of vertices 
within ${\mathcal R}$; they do not contribute any additional adjacency-lists.

We can bound $\mu-a-b$ using the observation that any intermediate visit of a tree node~$v$ on
${\mathcal T}$ is preceded by a last visit of a child $v'$ of $v$ and proceeded by a
first visit of another child $v''$ of $v$. Thus, $\mu-a-b \le \lceil \mu/2\rceil$, that is
$a+b \ge \lfloor \mu/2 \rfloor$, which implies that the expected number of distinct adjacency-lists being kept for ${\mathcal R}$ is at least $\lfloor \mu/2 \rfloor/2 \ge \mu/8$. 
\end{proof}

\section{The Dynamic Incremental Algorithm}
\label{algorithm}
In this section we concentrate on the incremental version for sparse graphs with 
$\Theta(n)$ updates where each update
inserts an edge. Thus, BFS levels can only decrease over time. 
Before we start, let us fix some notation:
for $i\ge 1$, $G_i=(V,E_i)$ is to denote the graph after the $i$-th update, 
$G_0$ is the initial graph. Let $d_i(v)$, $i\ge 0$, stand for the BFS level of node $v$ 
if it can be reached from the source node~$s$ in $G_i$ and $n$ otherwise. 
Furthermore, for $i \ge 1$, let $\Delta d_i(v)=| d_{i-1}(v) - d_i(v)|$. 
The main ideas of our approach are as follows: 

{\bf Checking Connectivity; Type A updates.}
In order to compute the BFS levels for $G_i$, $i \ge 1$, we first 
run an EM connected components algorithm (for example the one in \cite{MR_BFS}
taking $O(\sort(n) \cdot \log B)$ I/Os) in order to 
check, whether the insertion of the $i$-th edge~$(u,v)$ enlarges the connected 
component ${\mathcal C}_s$ of the source vertex~$s$. 
If yes (let us call this a {\em Type A update}), then w.l.o.g. let $u\in {\mathcal C}_s$
and let ${\mathcal C}_v$ be the connected component that comprises~$v$. 
The new edge $(u,v)$ is then the only connection between the existing BFS-tree for~$s$ and 
${\mathcal C}_v$. Therefore, we can simply run MR\_BFS on the subgraph $G'$ defined by
the vertices in ${\mathcal C}_v$ with source~$v$ and
add $d_{i-1}(u)+1$ to all distances obtained. This
takes $O(n_v+ \sort(n))$ I/Os where $n_v$ denotes the number of
vertices in ${\mathcal C}_v$.

If the $i$-th update does not merge ${\mathcal C}_s$ with some other connected component 
but adds an edge within ${\mathcal C}_s$ ({\em Type B update}) then we need to do something more fancy:

{\bf Dealing with small changes; Type B updates.} Now for computing the BFS levels
for $G_i$, $i \ge 1$, we pre-feed the adjacency-lists into a sorted pool ${\mathcal H}$
according to the BFS levels of their respective vertices in 
$G_{i-1}$ using a certain advance $\alpha >1$, i.e., the adjacency list
for $v$ is added to ${\mathcal H}$ when creating BFS level $\max\{0,d_{i-1}(v)-\alpha\}$ of $G_i$.
This can be done I/O-efficiently as follows. First we extract the adjacency-lists for vertices
having BFS levels up to $\alpha$ in $G_{i-1}$ and put them to ${\mathcal H}$ where they are kept sorted by node indices. From the remaining adjacency-lists we build a sequence~${\mathcal S}$ by
sorting them according to BFS levels in $G_{i-1}$ (primary criterion) and node indices 
(secondary criterion). For the construction of each new BFS level of $G_i$ we 
merge a subsequence of ${\mathcal S}$ accounting for one BFS level in $G_{i-1}$
with ${\mathcal H}$ using simple scanning. 
 
Therefore, if $\Delta d_i(v) \le \alpha$ for all $v \in V$ then all adjacency-lists will
be added to ${\mathcal H}$ in time and can be consumed from there without random I/O. 
Each adjacency-list is scanned at most once in ${\mathcal S}$ and at most
$\alpha$ times in ${\mathcal H}$. Thus, if 
$\alpha = o(\sqrt{B})$ this approach causes less I/O than MM\_BFS.
 
{\bf Dealing with larger changes.} 
Unfortunately, in general, there may be vertices~$v$ with $\Delta d_i(v) > \alpha$. 
Their adjacency-lists are not prefetched into ${\mathcal H}$ early enough and therefore have to
be imported into ${\mathcal H}$ using random I/Os to whole clusters
just like it is done in MM\_BFS. However, we apply the modified clustering
procedure described in Section~\ref{ssec:modified:preprocessing}  on $G_{i-1}$, the graph
without the $i$-th new edge (whose connectivity is the same as that of $G_i$) 
with chunk size~$\alpha/4$. 

Note that this may result in $\Theta(n/\alpha)$ cluster accesses, which would be prohibitive for
small~$\alpha$. Therefore we restrict the number of random cluster accesses to 
$\alpha \cdot n/B$. If the dynamic algorithm does not succeed within these bounds
then it increases $\alpha$ by a factor of two, computes a new clustering for $G_{i-1}$
with larger chunk size and starts a {\em new attempt} by repeating the whole approach with the increased parameters. Note that we do not need to recompute the spanning tree for the 
for the second, third, $\ldots$ attempt.

{\bf At most $O(\log B)$ attempts per update.} 
The $j$-th attempt, $j \ge 1$, of the dynamic approach
to produce the new BFS-level decomposition will apply an advance of 
$\alpha_j:=32 \cdot 2^{j}$ and recompute the modified clustering for $G_{i-1}$ using
chunk size $\mu_j:=8\cdot 2^j$.  
Note that there can be at most 
$O(\log \sqrt{B})=O(\log B)$ failing attempts for each edge update 
since by then our approach allows sufficiently
many random accesses to clusters so that all of them can be loaded explicitly
resulting in an I/O-bound comparable to that of static MM\_BFS.
In Section~\ref{analysis}, however, we will argue that for most edge updates within a longer sequence, 
the advance value and the chunk size value for the succeeding attempt are bounded 
by $O(B^{1/3})$ implying significantly improved I/O performance.

{\bf Restricting waiting time in ${\mathcal H}$.}
There is one more important detail to take care of: when adjacency-lists are brought
into ${\mathcal H}$ via explicit cluster accesses (because of insufficient advance $\alpha_j$ in the prefetching), these adjacency-lists will re-enter ${\mathcal H}$ once more later on during 
the (for these adjacency-lists by then useless) prefetching. Thus, in order to make sure
that unnecessary adjacency-lists do not stay in ${\mathcal H}$ forever, each entry in ${\mathcal H}$
carries a time-stamp ensuring that superfluous adjacency-lists are evicted from
${\mathcal H}$ after at most $\alpha_j=O(2^j)$ BFS levels.

\begin{lemma}
\label{lemma:single:typeB}
For sparse graphs with $O(n)$ updates, each Type B update succeeding during the
$j$-th attempt requires $O(2^j \cdot n/B + \sort(n) \cdot \log B)$ I/Os.
\end{lemma}
\begin{proof}
Deciding whether a Type B update takes place essentially requires a connected components
computation, which accounts for $O(\sort(n) \cdot \log B)$ I/Os. Within this
I/O bound we can also compute a spanning tree $T_s$ of the component holding the starting
vertex~$s$ but excluding the new edge. Subsequently, there
are $j = O(\log B)$ attempts, each of which uses $O(\sort(n))$ I/Os to derive a new
modified clustering based on an Euler tour with increasing chunk sizes around $T_s$.
Furthermore, before each attempt we need to initialize~${\mathcal H}$ and ${\mathcal S}$, which
takes $O(\sort(n))$ I/Os per attempt. The worst-case number of I/Os to (re-) scan 
adjacency-lists in ${\mathcal H}$ or to explicitly fetch clusters of adjacency-lists doubles
after each attempt. Therefore it asymptotically suffices to consider the (successful) last
attempt~$j$, which causes $O(2^j \cdot n/B)$ I/Os. Furthermore, each attempt
requires another $O(\sort(n))$ I/Os to pre-sort explicitly loaded clusters before they
can be merged with ${\mathcal H}$ using a single scan just like in MM\_BFS.
Adding all contributions yields the claimed I/O bound of $O(2^j \cdot n/B + \sort(n) \cdot \log B)$ for sparse graphs. 
\end{proof}

\section{Analysis}
\label{analysis}

We split our analysis of the incremental BFS algorithm into two parts.
The first (and easy one) takes care of Type A updates:

\begin{lemma}
For sparse undirected graphs with $\Theta(n)$ updates, there are at most $n-1$ Type A updates 
causing $O(n \cdot \sort(n) \cdot \log B)$ I/Os in total.
\end{lemma}
\begin{proof}
Each Type A update starts with an EM connected components computation causing
$O(\sort(n) \cdot \log B)$ I/Os per update. Since each node can be added to the
connected component ${\mathcal C}_s$ holding the starting vertex~$s$ only once, the total
number of I/Os spend in calls to the MR-BFS algorithm on components to
be merged with ${\mathcal C}_s$ is $O(n+\sort(n))$. Producing the output takes
another $O(\sort(n))$ per update. 
\end{proof}

Now we turn to Type B updates:

\begin{lemma}
For sparse undirected graphs with $\Theta(n)$ updates, all Type B updates 
cause $O(n \cdot (n^{2/3} + \sort(n) \cdot \log B))$ I/Os in total with high
probability.
\end{lemma}
\begin{proof}
Recall that $d_i(v)$, $i\ge 0$, stands for the BFS level of node $v$ 
if it can be reached from the source node~$s$ in $G_i$ and $n$ otherwise.
If upon the $i$-th update the dynamic algorithm issues an explicit fetch for
the adjacency-lists of some vertex~$v$ kept in some cluster $C$ then this is because
$\Delta d_i(v)=d_{i-1}(v) - d_i(v) > \alpha$ for the current advance~$\alpha$.
Note that for all other vertices $v'\in C$, there is a path of length at most $\mu$ in $G_{i-1}$, implying that $|d_{i-1}(v') - d_{i-1}(v)| \leq \mu$ as well as $|d_i(v) - d_i(v')| \leq \mu$. Having current chunk size $\mu=\alpha/4$, this implies
\begin{eqnarray*}
\Delta d_i(v')&=& d_{i-1}(v') - d_i(v')\\
 &=& d_{i-1}(v') - d_{i-1}(v) + d_{i-1}(v) - d_i(v) + d_i(v) - d_i(v')\\
 &>& \alpha-2 \mu \\
 &\ge& \alpha/2 . 
\end{eqnarray*}

If the $i$-th update needs $j$ attempts to succeed then, during the (failing) attempt
$j-1$, it has tried to explicitly access $\alpha_{j-1} \cdot n/B+1$ distinct clusters. 
Out of these at least $\alpha_{j-1} \cdot n/B=2^{j+4} \cdot n/B$ clusters carry 
an expected amount of at least $\mu_{j-1}/8=2^{j-1}$ adjacency-lists each. 
This accounts for an expected
number of at least $2^{2\cdot j+3} \cdot n/B$ distinct vertices, each of them
featuring $\Delta d_i(\cdot) \ge \alpha_{j-1}/2=2^{j+3}$.
With probability at least $1/2$ we actually get at least half
%
of the expected
amount of distinct vertices/adjacency-lists, i.e., $2^{2\cdot j+2} \cdot n/B$.
Therefore, using the definitions $D_i=\sum_{v \in V \setminus \{ s\}} d_i(v)$ and $\Delta D_i= |D_{i-1} -D_i|$, if the $i$-th update succeeds within the
$j$-th attempt we have $\Delta D_i \ge 2^{3\cdot j+5} \cdot n/B =: Y_j$ with probability at
least $1/2$. 
Let us call this event {\em a large $j$-yield}. 

Since each attempt uses a new clustering with independent choices for $r(\cdot)$,
if we consider two updates~$i'$ and~$i''$ that succeed after the same number of
attempts~$j$, then both $i'$ and $i''$ have a large yield with probability at 
least $1/2$, independent of each other.
Therefore, we can use Chernoff bounds \cite{79799} in order to show that out
of $k \ge 16 \cdot c \cdot \ln n$ updates that all succeed within their $j$-th
attempt, at least $k/4$ of them have a large $j$-yield with probability at least
$1-n^{-c}$ for an arbitrary positive constant~$c$.
Subsequently we will prove an upper bound on the total number of large $j$-yields that
can occur during the whole update sequence.

The quantity $\Delta D_i$ provides a global measure as for how much the BFS levels change after 
inclusion of the $i$-th edge from the update sequence. 
If there are $m'=\Theta(n)$ edge inserts in total, then
\[ n^2 > D_0 \ge D_1 \ge \ldots \ge D_{m'-1} \ge D_{m'} > 0.\]
A large $j$-yield means $\Delta D_i \ge Y_j$. Therefore, in the
worst case there are at most $n^2/Y_j=n^2/(2^{3\cdot j+5} \cdot n/B)=n\cdot B /2^{3\cdot j+5}$
large $j$-yield updates and -- according to our discussion above -- it needs at most
$k_j:=4 \cdot n\cdot B /2^{3\cdot j+5}$ updates that succeed within the $j$-th attempt to 
have at least $k_j/4$ large $j$-yield updates with high 
probability\footnote{We also need to verify that $k_j \ge 16 \cdot c \cdot \ln n$.
As observed before, the dynamic algorithm will not increase its advance and chunk 
size values beyond $O(\sqrt{B})$ implying $2^j=O(\sqrt{B})$. But then we have
$k_j=4 \cdot n \cdot B/2^{3 \cdot j+5}=\Omega(n/\sqrt{B})$ and 
$n/\sqrt{B}\ge n/\sqrt{M} \ge n/\sqrt{n} \ge 16 \cdot c \cdot \ln n$ for sufficiently
large~$n$.}.

For the last step of our analysis we will distinguish two kinds of Type~B updates:
those that finish using an advance value~$\alpha_{j^*} < B^{1/3}$ (Type~B1), and
the others (Type~B2). Independent of the subtype, an update costs
$O(\alpha_{j^*} \cdot n /B + \sort(n) \cdot \log B)$ $=$
$O(2^{j^*} \cdot n /B + \sort(n) \cdot \log B)$ I/Os by 
Lemma~\ref{lemma:single:typeB}.
Obviously, for an update sequence of $m'=\Theta(n)$ edge insertions
there can be at most $\Theta(n)$ updates of Type~B1, each of them
accounting for at most $O( n /B^{2/3} + \sort(n) \cdot \log B)$ I/Os.
As for Type~B2 updates we have already shown that with high probability
there are at most $O(n \cdot B/2^{3 \cdot j^*})$ updates that succeed with
advance value $\Theta(2^{j^*})$. Therefore, using Boole's inequality, the total
amount of I/Os for all Type~B2 updates is bounded by
\[O\left( \left(\sum_{g\ge 0} 
\frac{n \cdot B}{(B^{1/3} \cdot 2^g)^3} \cdot
\frac{B^{1/3} \cdot 2^g \cdot n}{B} \right) + n \cdot \sort(n) \cdot \log B  \right)=\] 
$O(n \cdot (n/B^{2/3} + \sort(n) \cdot \log B ))$ with high probability. 
\end{proof}

Combining the two lemmas of this section implies 

\begin{theorem}
For general sparse undirected graphs of initially $n$ nodes and $O(n)$ edges
and $\Theta(n)$ edge insertions, dynamic BFS can be solved using 
amortized $O(n/B^{2/3}+\sort(n)\cdot \log B)$ I/Os per update with high
probability. 
\end{theorem}

\section{Decremental Version and Extensions.}
\label{extensions}
Having gone through the ideas of the incremental version, it is now close to trivial
to come up with a symmetric external-memory dynamic BFS algorithm for a sequence of edge
deletions: instead of pre-feeding adjacency-lists into 
using an {\em advance} of $\alpha_j$ levels, we now apply a {\em lag} of $\alpha_j$ levels.
Therefore, the adjacency-list for a vertex~$v$ is found in ${\mathcal H}$ as long as the 
deletion of the $i$-th edge does not increase $d_i(v)$ by more than~$\alpha$. Otherwise,
an explicit random access to the cluster containing $v$'s adjacency-list is issued later on.
All previously used amortization arguments and bounds carry through, the only difference being
that $d_i(\cdot)$ values may monotonically increase instead of decrease.  

Better amortized bounds can be obtained if $\omega(n)$ updates take place and/or
$G_0$ has $\omega(n)$ edges. Then we have the potential to amortize 
more random accesses per attempt, which leads to larger $j$-yields and
reduces the worst-case number of expensive updates. Consequently, we can reduce the
defining threshold between Type~B1 and Type~B2 updates, thus eventually 
yielding better amortized I/O bounds. Details we be provided in the full version
of this paper. 

Modifications along similar lines are in order if external-memory
is realized by flash disks~\cite{GalToledoSURVEY}: compared to hard disks, flash memory can sustain many more unstructured read I/Os per second but on the other hand flash memory usually offers less read/write bandwidth than hard disks. Hence, in algorithms like ours 
that are based on a trade-off between unstructured read I/Os and bulk 
read/write I/Os, performance can be improved by allowing more unstructured read I/Os
(fetching clusters) if this leads to less overall I/O volume (scanning hot pool entries).

\section{Conclusions}
\label{conclusions}
We have given the first non-trivial external-memory algorithm 
for dynamic BFS. Even though we obtain significantly
better I/O bounds than for the currently best static algorithm, there
are a number of open problems: first of all, our bounds dramatically
deteriorate for mixed update sequences (edge insertions and edge deletions
in arbitrary order and proportions); besides oscillation effects, a single
edge deletion (insertion) may spoil a whole chain of amortizations for
previous insertions (deletions).
Also, it would be interesting to see, whether our bounds can be
further improved or also hold for shorter update sequences.
Finally, it would be nice to come up with a deterministic version
of the modified clustering.

\subsection*{Acknowledgements}
We would like to thank Deepak Ajwani for very helpful discussions.

\bibliographystyle{plain}

\begin{thebibliography}{10}

\bibitem{AGG_VIT}
A.~Aggarwal and J.~S. Vitter.
\newblock The input/output complexity of sorting and related problems.
\newblock {\em Communications of the ACM, 31(9)}, pages 1116--1127, 1988.

\bibitem{Ext_MST}
L.~Arge, G.~Brodal, and L.~Toma.
\newblock On external-memory {MST}, {SSSP} and multi-way planar graph
  separation.
\newblock In {\em Proc. 8th Scand. Workshop on Algorithmic Theory (SWAT)},
  volume 1851 of {\em LNCS}, pages 433--447. Springer, 2000.

\bibitem{Euler}
M.~Atallah and U.~Vishkin.
\newblock Finding {E}uler tours in parallel.
\newblock {\em Journal of Computer and System Sciences, 29(30)}, pages
  330--337, 1984.

\bibitem{BGVW00}
A.~Buchsbaum, M.~Goldwasser, S.~Venkatasubramanian, and J.~Westbrook.
\newblock On external memory graph traversal.
\newblock In {\em Proc. 11th Ann. Symposium on Discrete Algorithms (SODA)},
  pages 859--860. ACM-SIAM, 2000.

\bibitem{Chi_Ext_gr_alg}
Y.~J. Chiang, M.~T. Goodrich, E.~F. Grove, R.~Tamasia, D.~E. Vengroff, and
  J.~S. Vitter.
\newblock External memory graph algorithms.
\newblock In {\em Proc. 6th Ann.Symposium on Discrete Algorithms (SODA)}, pages
  139--149. ACM-SIAM, 1995.

\bibitem{Cormen}
T.~H. Cormen, C.E. Leiserson, and R.L. Rivest.
\newblock {\em Introduction to Algorithms}.
\newblock McGraw-Hill, 1990.

\bibitem{Tradeoff_Streaming}
C.~Demetrescu, I.~Finocchi, and A.~Ribichini.
\newblock Trading off space for passes in graph streaming problems.
\newblock {\em In 17th ACM-SIAM Symposium on Discrete Algorithms}, pages
  714--723, 2006.

\bibitem{EppGalIta-ATCH-99}
D.~Eppstein, Z.~Galil, and G.~Italiano.
\newblock {Dynamic graph algorithms}.
\newblock In Mikhail~J. Atallah, editor, {\em Algorithms and Theory of
  Computation Handbook}, chapter~8. CRC Press, 1999.

\bibitem{GalToledoSURVEY}
E.~Gal and S.~Toledo.
\newblock Algorithms and data structures for flash memories.
\newblock {\em ACM Computing Surveys}, 37:138--163, 2005.

\bibitem{79799}
T.~Hagerup and C.~R\"{u}b.
\newblock A guided tour of chernoff bounds.
\newblock {\em Inf. Process. Lett.}, 33(6):305--308, 1990.

\bibitem{MM_BFS}
K.~Mehlhorn and U.~Meyer.
\newblock External-memory breadth-first search with sublinear {I/O}.
\newblock In {\em Proc. 10th Ann. European Symposium on Algorithms (ESA)},
  volume 2461 of {\em LNCS}, pages 723--735. Springer, 2002.

\bibitem{Alg_Mem}
U.~Meyer, P.~Sanders, and J.~Sibeyn~(Eds.).
\newblock {\em Algorithms for Memory Hierarchies}, volume 2625 of {\em LNCS}.
\newblock Springer, 2003.

\bibitem{MR_BFS}
K.~Munagala and A.~Ranade.
\newblock {I/O}-complexity of graph algorithms.
\newblock In {\em In Proc. 10th Ann. Symposium on Discrete Algorithms (SODA)},
  pages 687--694. ACM-SIAM, 1999.

\bibitem{Roditty:2006}
L.~Roditty.
\newblock {\em Dynamic and static algorithms for path problems in graphs}.
\newblock PhD thesis, Tel Aviv University, 2006.

\bibitem{Ullman91}
J.~D. Ullman and M.~Yannakakis.
\newblock High-probability parallel transitive closure algorithms.
\newblock {\em SIAM Journal on Computing}, 20(1):100--125, February 1991.

\bibitem{VIT1}
J.~S. Vitter.
\newblock External memory algorithms and data structures: Dealing with massive
  data.
\newblock {\em ACM computing Surveys, 33}, pages 209--271, 2001.
\newblock Revised version (August 2007) available online at
  \texttt{http://www.cs.purdue.edu/homes/jsv/Papers/Vit.IO\_survey.pdf}.

\end{thebibliography}

\end{document}